\documentclass[11pt,a4paper,english]{article}

\title{Quantum Mechanical Effects from Deformation Theory} 
\author{A. Much\\ \footnotesize{Max-Planck-Institute for Mathematics in the Sciences,
04103 Leipzig, Germany}
\\ \footnotesize{Institute for Theoretical Physics, University of Leipzig, 04009 Leipzig,
Germany}}

\usepackage[pdftex]{graphicx} 
\usepackage{float}            
\usepackage[]{babel}
\usepackage[T1]{fontenc}   
\usepackage[cmex10]{amsmath}  
\usepackage{amstext}          
\usepackage{mathrsfs}
\usepackage{amsfonts}         
\usepackage{amssymb}          
\usepackage{setspace}
\usepackage{bm}               
\usepackage{enumerate}        
\usepackage[bottom]{footmisc} 
\usepackage{array}            
\usepackage{textcomp}         
\usepackage{pdfpages}         
\usepackage{parskip}          
\usepackage[right]{eurosym}   
\usepackage{xcolor}           
\usepackage[square,numbers]{natbib}	  
 
\usepackage[hyphens]{url}     
\usepackage{multicol}         
\usepackage{graphicx}
\usepackage{subfigure}
\usepackage{dsfont}
\usepackage{amsmath}
\usepackage{amssymb}
\usepackage{amsthm}
\usepackage{moreverb}
\usepackage[utf8]{inputenc}
\usepackage{tabularx}
\usepackage{longtable}
\usepackage{nicefrac}
\usepackage[pdf]{pstricks}
\usepackage{tikz}
\usepackage[off]{auto-pst-pdf}
\usepackage[colorlinks,pdfpagelabels,pdfstartview = FitH,bookmarksopen = true,bookmarksnumbered =
true,linkcolor = black,plainpages = false,hypertexnames = false,citecolor = black]{hyperref}

\newtheorem{theorem}{\textsc{Theorem}}[section]
\newtheorem{lemma}{\textsc{Lemma}}[section]
\newtheorem{proposition}{\textsc{Proposition}}[section]

\theoremstyle{definition}
\newtheorem{definition}{\textsc{Definition}}[section]

\theoremstyle{remark}
\newtheorem{remark}{Remark}[section]

\usepackage{fancyhdr}
 \numberwithin{equation}{section} 
\begin{document}
\maketitle
\abstract{We consider deformations of quantum mechanical operators by using the
novel construction tool of warped convolutions. The deformation enables us to
obtain several quantum
mechanical effects where electromagnetic and gravitomagnetic fields play a role. 
Furthermore, a quantum plane can be defined by using the deformation techniques. This in turn gives an experimentally verifiable effect.  } 
  \tableofcontents

\section{Introduction}
Deformation theory is an interesting subject of research, both from a mathematical and  a
physical point of view. Nowadays, many fundamental theories are reconsidered as deformations of
more subtle theories. A fundamental example of deformation theory in a
physical context is the deformation  of classical mechanics to quantum physics, where the
deformation parameter in that case is Planck's constant $\hbar$. 
In this context the Poincar\'e group can also be considered as a deformation of
the Galilei group, where the  parameter characterizing the deformation is given by the speed of
light, i.e. $1/c^2$. The opposite of a deformation in a group theoretical context is a
contraction. It is induced by taking the limit of the deformation parameter   to zero. In
the example of the Poincar\'e group, this would mean that we take the limit $1/c^2\rightarrow 0 $.
This limit is often taken by physicists  as a consistency check and rarely recognized  as a contraction. Another interesting example is the deformation of the Poincar\'e group to the Anti-de Sitter group by using 
the cosmological constant $\Lambda$ as a deformation parameter. \newline\newline
Thus, from a physical point of view, deformation theory enters the game by the physical
dimensionality of the deformation parameter. In this work, we emphasize the importance of choosing
the \textbf{deformation constant}, in order to obtain \textbf{physical effects}.
\newline\newline
One justified critique usually spoken out in the context of deformation theory  is that the rightful deformation is only
guessed \textbf{after} the physical theory has been formulated. Thus, to consider such deformations as
fundamental, is often put into the category of wishful thinking of theoretical physicists.
Therefore, the main aim of the current work is to understand  a variety of physical effects,
in a quantum mechanical context, by a deformation of the free theory. Furthermore, we propose an effect coming from
deformation considerations.
\newline\newline
The method that is used, in the current work, for deformation is known under the name of warped convolutions,  \cite{GL1, BS, BLS}.  Usually, this method
is used in the realm of quantum field theory to deform free quantum fields and to construct
non-trivial interacting fields which was done in \cite{A, GL1, GL2, GL4, GL5, Mor, MUc}. It was also used in quantum measurement theory \cite{AA}. One of the major advantages of this method   is its easy accessibility to a physical regimen. \newline\newline By using this novel tool in a quantum mechanical context, we recast many fundamental physical effects involving electromagnetism. This is done by the adjustment of the deformation parameter.  Moreover, we are able to produce gravitomagnetic effects and interaction between magnetic and gravitomagnetic fields by this  deformation procedure. \newline\newline
This paper is organized as follows: In Section \ref{s2} we give a brief introduction of the method of warped convolutions and introduce the basic notations for  deformation in a quantum mechanical context. The free Hamiltonian is deformed in Section \ref{s3}. We are obliged to show that the warped convolutions formula, originally formulated for a subset of bounded operators, is well-defined in the case of the deformation of unbounded operators. Section \ref{s4} is devoted to the emergence of physical effects from the deformation procedure.

 \section{Warped convolutions in QM}\label{s2}
   Since we constantly use warped convolutions we lay out the novel deformation procedure  in this section
and present the most important definitions, lemmas and propositions for the current paper.
For proofs of the lemmas and propositions we refer the reader to the original works. 
\newline\newline 
We start by assuming the
existence of a strongly continuous unitary group $U$ that is a representation of the additive
group $\mathbb{R}^n$, on some separable Hilbert space $\mathscr{H}$. 
Let $\mathcal{D}$ be the dense domain of vectors in $\mathscr{H}$ which transform
smoothly under the adjoint action of $U$.
Then, the warped convolutions  for operators  $F\in C^{\infty}$, where $C^{\infty}$ 
is the *-algebra of smooth elements with respect to the adjoint  action of $U$, are given by
the following definition.\\
\begin{definition}
 Let $B$ be a real skew-symmetric matrix on $\mathbb{R}^n$, let $F\in C^{\infty}$ and let $E$ be
the spectral resolution of the unitary operator $U$.  Then, the corresponding warped convolution $F_{B}$
of $F$ is defined on the domain $\mathcal{D}$ according to
\begin{equation}\label{WC}
 F_{B}:=\int \alpha_{Bx}(F)dE(x),
\end{equation}
where $\alpha$ denotes the adjoint action of $U$ given by $\alpha_{k}(F)=U(k)\,F\,U(k)^{-1}$.
\end{definition}$\,$\newline
The restriction in the choice of operators is owed to the fact that the deformation is performed
with operator valued integrals. Furthermore,  one can represent the warped
convolution of $
 {A} \in \mathcal{C}^{\infty}$    by $\int \alpha_{B x}(A) dE(x)$ or $\int dE(x)
\alpha_{B x}(A) $, on the dense domain
$\mathcal{D}\subset\mathscr{H}$ of vectors smooth w.r.t. the action of $U$,  in terms of
strong limits 
\begin{equation*}
\int\alpha_{B x}(A) dE(x)\Phi=(2\pi)^{-n}
\lim_{\epsilon\rightarrow 0}
\int\int d^{n}x\, d^{n}y \,\chi(\epsilon x,\epsilon y )\,e^{-ixy}\, U(y)\, \alpha_{B x}(A)\Phi,  
\end{equation*}
where $\chi \in\mathscr{S}(\mathbb{R}^n\times\mathbb{R}^n)$ with $\chi(0,0)=1$.
This representation makes calculations and proofs concerning the existence of  integrals
easier. In this work we  use both representations. \newline\newline 
In the following lemma we introduce the deformed product, also known as the Rieffel product
\cite{RI} 
by using warped convolutions. The two deformations are interrelated since   warped convolutions supply isometric
representations of Rieffel's strict deformations of $C^{*}$-dynamical systems with
actions of $\mathbb{R}^n$.
 \newline
\begin{lemma}\label{l2.1}
Let $B$ be a real skew-symmetric matrix on $\mathbb{R}^n$ and let $  {A},   {E} \in
\mathcal{C}^{\infty}$. Then 
\begin{equation*}
 {A}_{B}   {E}_{B} \Phi= (A\times_{B }E)_{B}\Phi, \qquad
\Phi\in\mathcal{D}.
\end{equation*}
where $\times_{B}$ is known as the Rieffel product
on $\mathcal{C}^{\infty}$ and is given by, 
\begin{equation}\label{dp0}
(A\times_{B}E )\Phi=(2\pi)^{-n}
\lim_{\epsilon\rightarrow 0}
\int\int d^{n}x\, d^{n}y \,\chi(\epsilon x,\epsilon y )\,e^{-ixy} \, \alpha_{B x}(A)\alpha_{y}(E)\Phi.
\end{equation}
\end{lemma}$\,$\newline
Another proposition that seems a matter of technicality in the original work but has great
physical significance  is the following.
\begin{proposition}\label{wc2}
 Let $B_1$, $B_2$ be skew symmetric matrices. Then 
\begin{equation}
 \left(A_{B_{1}}\right)_{B_{2}}=A_{B_{1}+B_{2}},\qquad A\in C^{\infty}.
\end{equation}
\end{proposition}$\,$\newline
Next, we adopt Formula (\ref{WC}) to define the warped convolutions for an unbounded
operator,  with a  real vector-valued function of the coordinate operator. To apply the definition of   warped convolutions, we need   self-adjoint operators that commute along their components. For this purpose let us give the following theorem, \cite[ Theorem VIII.6]{ RS1}.

\begin{theorem}\label{tsa}
Let $\mathbf{Q}(.)$ be an unbounded real vector-valued Borel function on  $\mathbb{R}^n $ and let the dense domain $D_{\mathbf{Q}}$ be given as, 
\begin{equation*}
D_{\mathbf{Q}}=\{\phi| \int\limits_{ -\infty}^{\infty} |Q_j(\mathbf{x})|^2\,d(\phi,P_\mathbf{x}\phi) <\infty, \, \quad j=1,\dots, n\},
\end{equation*}
where $\{P_x\}$ are projection valued measures on $\mathscr{H}$. Then, $\mathbf{Q}(\mathbf{X})$ defined on $D_{\mathbf{Q}}$ is a self-adjoint operator.

\end{theorem} $\,$\newline
In this paper we  consider unbounded real vector-valued functions of the \textbf{coordinate operator} and therefore we give the following definition. 
\begin{definition}  
Let $B$ be a real skew-symmetric matrix on $\mathbb{R}^{n}$ and 
let $\chi \in\mathscr{S}(\mathbb{R}^n\times\mathbb{R}^n)$ with $\chi(0,0)=1$. Moreover,  let 
$\mathbf{Q}(\mathbf{X})$ be given as in Theorem \ref{tsa}. Then, the warped convolutions of an   operator $A$ with  operator $\mathbf{Q}$, denoted as 
$A_{B,\mathbf{Q} }$  are
defined, in the same manner as in \cite{BLS}, namely
\begin{equation}\label{defcop}
A_{B,\mathbf{Q}  } := (2\pi)^{-n}
\lim_{\epsilon\rightarrow 0}
  \iint  \, d^{n}y \,  d^{n}k \, e^{-iy_{l}k^{l}}  \, \chi(\epsilon y,\epsilon k)V(k)
\alpha_{B y}(A) . 
\end{equation}
The automorphisms $\alpha$ are implemented by the adjoint action of the strongly
continuous unitary representation $V(y)=e^{iy_kQ^k}$ of $\mathbb{R}^n$ given by  
\begin{equation*}
 \alpha_{y}(A)=V(y)\,A\,V(y)^{-1}, \quad y \in \mathbb{R}^n.
\end{equation*} 
\end{definition}$\,$\newline
Since we deform unbounded operators we are obliged to prove  that the deformation
formula, given as an oscillatory integral, is well-defined.  This is the subject of the next section.

\section{Deforming Unbounded Operators}\label{s3}
At first, we study deformations of the
simplest Hamiltonian of quantum mechanics, that of a free particle.  Further on we explore the
physical consequences of the deformation and introduce to the reader how one can obtain a variety
of physical effects using this method. For a deformation of the Hamiltonian we choose to work in the
standard realization of quantum mechanics, the so called \textbf{Schr\"odinger
representation}, \cite{BEH, RS1, T}. In this representation the pair of operators
$(P_{i},X_{j})$, satisfying the \textbf{canonical commutation relations}  (CCR)
\begin{equation}\label{ccr}
[X_i,P_{k}]=-i\eta_{ik},
\end{equation}
are represented as essentially self-adjoint operators on the dense domain
$\mathscr{S}(\mathbb{R}^n)$. Here $X_{i}$ and $P_{k}$ are the closures of
$x_{i}$ and multiplication by $i {\partial}/{\partial x^k}$ on
$\mathscr{S}(\mathbb{R}^n)$ respectively. 
\newline\newline
In quantum mechanics  the Hamiltonian of a free
particle
is given as follows
\begin{equation}\label{fh}
H_{0}= -\frac{P_jP^j}{2m}.
\end{equation}
This operator describes a non-relativistic and non-interacting particle. For the following considerations, we restrict the
deformation to three space dimensions. This restriction is obvious due to its physical relevance.
Let us
start this section with a  theorem concerning the domain of self-adjointness
and the spectrum of  the free undeformed Hamiltonian $H_{0}$, \cite{T}.
\newline  
\begin{theorem}
 The free Schr\"odinger operator $H_{0}$ is self-adjoint on the domain
$\mathcal{D}(H_{0})$ given as 
\begin{equation*}
\mathcal{D}(H_{0})=H^2(\mathbb{R}^3)=\{\varphi \in
L^2(\mathbb{R}^3)||\textbf{P}|^2\varphi\in
L^2(\mathbb{R}^3)\},
\end{equation*}
and its spectrum is characterized by $\sigma(H_{0})=[0,\infty)$.
\end{theorem}$\,$\newline
Before proceeding with the deformation, one problem arises at this point of our work.
The deformation formula  given by warped convolutions is only well-defined in the strong operator
topology for a subset of \textbf{bounded operators} that are smooth w.r.t. the unitary representation
$U$ of
$\mathbb{R}^n$. In view of the fact that we deal with unbounded operators, we are obliged to investigate the validity of the deformation Formula (\ref{defcop})
for $H_{0}$.  For this purpose we need a dense domain $ \mathcal{E}\subseteq \mathscr{S}(\mathbb{R}^3)$  that fulfills additional requirements. 
\begin{lemma}\label{cnt0}
Consider the self-adjoint operator 

\begin{equation}\label{hs0}
 \mathbf{Q}(\mathbf{X}) =\mathbf{X}/ \vert\mathbf{X}\vert^{ n}, \qquad    n\in\mathbb{R}.
\end{equation}
 Then,  for all $n\in\mathbb{R}$  there exists a dense domain $ \mathcal{E}\subseteq \mathscr{S}(\mathbb{R}^3)$ such that
\begin{align} \label{hs}
\Vert 
\{P_{j},[\mathbf{Q},P^j]\}  \Phi \Vert<\infty, \qquad 
\Vert 
[\mathbf{Q},P_j][\mathbf{Q},P^j]  \Phi \Vert<\infty, \qquad \Phi\in  \mathcal{E} .
\end{align}
\end{lemma}
\begin{proof}
From Theorem \ref{tsa} it is follows that all operators of the form $\mathbf{X}/ \vert\mathbf{X}\vert^{ n}$ are self-adjoint on their respective domains.  Further we show the existence of a dense domain, satisfying Inequalities (\ref{hs}).  To simplify calculations let us give  general formulas for the commutators 
\begin{align}\label{hs1}
[P_j,  \vert\mathbf{X}\vert^{-n}]= i \,n\, 
 X_j\vert\mathbf{X}\vert^{ -(n+2)} ,
\end{align}
\begin{align}\label{hs2}
[P_j,   X_k /\vert\mathbf{X}\vert^{n}]= i\left(\eta_{jk}+
 n\,X_kX_j/\vert\mathbf{X}\vert^{ 2}\right)\vert\mathbf{X}\vert^{-n}.
\end{align} 
Thus, for an arbitrary  $n\in\mathbb{R}$  and $\mathbf{Q} =\mathbf{X}/ \vert\mathbf{X}\vert^{ n}$  the anti-commutator in Inequality (\ref{hs}) is calculated as follows
\begin{align*} 
\{P_{j},[P^j, {Q}^k ]\} 
&=[P_j,[P^j, {Q}^k ]]+2[P^j, {Q}^k ]P_{j}\\&=
i[P_j,\left(\eta^{jk}+
 n X^kX^j/\vert\mathbf{X}\vert^{ 2}\right)\vert\mathbf{X}\vert^{-n} ]+2i\left(\eta^{jk}+
 n X^kX^j/\vert\mathbf{X}\vert^{ 2}\right)\vert\mathbf{X}\vert^{-n}P_{j}\\&
=  \underbrace{ \left(n^2-3n\right) }_{=:a(n)}
 X^k \vert\mathbf{X}\vert^{ -(n+2)}  +
2i\left(\eta^{jk}+
 n\,X^kX^j/\vert\mathbf{X}\vert^{ 2}\right)\vert\mathbf{X}\vert^{-n}P_{j},
\end{align*}where in  the last lines we used the CCR, Equations (\ref{hs1}) and (\ref{hs2}).
The norm of the anti-commutator is given by
\begin{align*}
 \Vert 
\{P_{j},[\mathbf{Q},P^j]\} \Phi \Vert &=  \Vert e_k\left(a(n)\,
 X^k \vert\mathbf{X}\vert^{ -(n+2)} +
2i\left(\eta^{jk}+
 n\,X^kX^j/\vert\mathbf{X}\vert^{ 2}\right)\vert\mathbf{X}\vert^{-n}P_{j} \right)\Phi \Vert \\&\leq
\Vert e_k\,a(n)\,
 X^k \vert\mathbf{X}\vert^{ -(n+2)} \Phi \Vert+ \Vert
2  e_k\left(\eta^{jk}+  
 n\,X^kX^j/\vert\mathbf{X}\vert^{ 2}\right)\vert\mathbf{X}\vert^{-n}P_{j} \Phi \Vert  \\&\leq 
\Vert a(n)
 \vert\mathbf{X}\vert^{ -(n+1)} \Phi \Vert+\Vert
2   \,\vert\mathbf{X}\vert^{-n} \mathbf{P} \Phi \Vert + \Vert 
 2 n\, \vert\mathbf{X}\vert^{-(n+1) }\,X^j  P_{j}\Phi \Vert.
\end{align*}
 The term in the second inequality in (\ref{hs})  is given by
\begin{align*}
\Vert 
[\mathbf{Q},P_j][\mathbf{Q},P^j]  \Phi \Vert&=\Vert 
\left(\eta_{jl}+
 n\,X_lX_j/\vert\mathbf{X}\vert^{ 2}\right)
 \left(\eta^{jl}+
 n\,X^lX^j/\vert\mathbf{X}\vert^{ 2}\right) \vert\mathbf{X}\vert^{-2n} \Phi \Vert\\&=
\Vert  (n^2-2n+3) \vert\mathbf{X}\vert^{-2n}
\Phi \Vert.
\end{align*}
It is clear that if  $n\in \mathbb{R}^{-}_0$ Inequalities (\ref{hs}) are satisfied for vectors in the dense domain $\mathscr{S}(\mathbb{R}^3)$, since the expressions in the norm are positive polynomial functions of the coordinate operator. 
For $n\in \mathbb{R}^{+}$ we consider  the domain $\mathcal{E}$  which denotes the linear hull of the \textit{dense} vectors \cite[Theorem 3.2.5]{Th}
\begin{align*}
 \Phi(\mathbf{x})=x_1^{k_1}x_2^{k_2}x_3^{k_3}\,\exp{(-\frac{|\mathbf{x}|^2}{2} )},\qquad k_i=0,\,1,\,2,\,\dots.
\end{align*} 
Since the dense domain $\mathcal{E}$ remains invariant under the action of positive functions of the coordinate and momentum operator (see proof of \cite[Theorem 3.2.5]{Th}), the remaining task is to show  the finiteness of 
\begin{align*} 
\Vert 
 \vert\mathbf{X}\vert^{  (\lambda-1)} \Phi \Vert ^2&=\int d^3 \mathbf{x}\,|\mathbf{x}|^{2(\lambda-1)}\,e^{-|\mathbf{x}|^2}, \qquad \lambda\in\mathbb{R}^{-}  \\&
=\int\limits_0^{2\pi}d\phi\int\limits_0^{\pi}d\vartheta \sin\vartheta\int\limits_0^{\infty}r^{2\lambda}\,e^{-r^2}dr\\&= 2\pi  \int\limits_{-\infty}^{\infty}x^{2\lambda}\,e^{-x^2}dx.
\end{align*} This integral exists for all $\lambda$ and it is easily seen to be an analytic function in $\lambda$, 
\cite[Chapter 1, Section 3.6]{GS1}.
Note that we choose the polynomial functions of components of  $\mathbf{x}$ to be equal to one. This choice is for the sake of argument, since  positive polynomial functions improve the behavior of the integral. 
\end{proof}$\,$\newline
 By using the former lemma,  we show in the next proposition that the scalar product of the  deformed free Hamiltonian, i.e. 
\begin{align*}
\langle \Psi, (H_{0})_{B,\mathbf{Q}}\Phi\rangle&=
(2\pi)^{-3} {
\lim_{\varepsilon\rightarrow 0}
  \iint  \, d^{3}y \,  d^{3}k \, e^{-iy_{l}k^{l}}  \, \chi(\varepsilon y,\varepsilon
k)\langle \Psi, 
V(k)\alpha_{By}(H_0)\Phi\rangle},
\end{align*}
is \textbf{bounded} for   $ \forall  \Psi\in \mathscr{H}$ and $\Phi \in \mathcal{E}\subseteq \mathscr{S}(\mathbb{R}^3) $. \newline

\begin{proposition}\label{wcfh}
Let 
$\mathbf{Q}(\mathbf{X})$ be a self-adjoint  operator of the form
\begin{align*}
 \mathbf{Q}(\mathbf{X}) =\mathbf{X}/ \vert\mathbf{X}\vert^{ n}, \qquad    n\in\mathbb{R},
\end{align*}
and let $(H_{0})_{B,\mathbf{Q}}$ denote the deformed free Hamiltonian (see Formula (\ref{defcop})). Then, the scalar product $\langle \Psi, (H_{0})_{B,\mathbf{Q}}\Phi\rangle$ is bounded by a finite constant $C_B$ as follows,
\begin{align*}\vert
\langle \Psi, (H_{0})_{B,\mathbf{Q}}\Phi\rangle\vert&\leq
C_B\,\Vert \Psi\Vert, \qquad  \forall  \Psi\in \mathscr{H},\,\, \Phi \in \mathcal{E}\subseteq \mathscr{S}(\mathbb{R}^3).
\end{align*}
Therefore, the deformation formula  for the unbounded operator
$H_{0}$, given as an oscillatory
integral, is   well-defined and the explicit result of the deformation is  
\begin{equation}\label{defh}
(H_{0})_{B,\mathbf{Q}}\Phi=
-\frac{1}{2m}\left
(P_j+i(B  Q )^k [Q_k,P_j]  \right)\left(P^j+i(B  Q )^r [Q_r,P^j]  \right)
\Phi.
\end{equation}
 
\end{proposition}

\begin{proof}
To prove the boundedness of the  scalar product  ${\langle \Psi, 
V(k)\alpha_{By}(H_0)\Phi\rangle}$, we first derive the adjoint action of $V(B y )$ on $H_{0}$  given
by,\begin{align*}
\alpha_{By}(H_0)&=-\frac{1}{2m}V(By)P_jP^jV(-By)\\&=-\frac{1}{2m}V
(By)P_jV(-By)V(By)P^jV(-By).
\end{align*}
To solve this expression we first   calculate the adjoint action of $V(By)$ on the
momentum operator $P_j$ by using  the Baker-Campbell-Hausdorff formula,  
\begin{align}\label{aaxp} 
V(By)P_jV(-By)& 
 =P_j+i(By)^k\underbrace{[Q_k,P_j]
}_{=:-iX_{kj} (\mathbf{X})}
 + 
\underbrace{\frac{i^2}{2}(By)^l (By)^k [Q_l,[Q_k,P_j]
]+...}_{=0},
\end{align}
 where in the last lines we used the CCR given in (\ref{ccr}),  and the commutativity of the coordinate  operator, i.e. $[X_{i},X_{j}]=0$.
Thus, the adjoint action w.r.t. $V (B y )$ on $H_{0}$ is
\begin{align*}\label{roids2}
\alpha_{By}(H_0)&=
-\frac{1}{2m}  ( 
P_{j}+(By)^s  X_{sj}     )  (
P^{j}+(By)_r X^{rj}  )
 \nonumber \\&=
H_{0} -(By)^s\underbrace{\frac{1}{2m}\left(
P^{j} X_{sj} + 
 X_{sj}P^{j}\right)}_{=:N_s} -(By)_r(By)^s \underbrace{ \frac{1}{2m}X_{sj} X^{rj}}_{=:R_s^{\,\,r}}
 \nonumber \\&= H_{0}-(By)^s N_s-(By)_r(By)^s R_s^{\,\,r}.
\end{align*}
 Moreover, without loss of generality, one can choose the  skew-symmetric matrix $B$  to have the  
form  $B_{ij}= \varepsilon_{ijk}B^j$, where $\varepsilon_{ijk}$ is the three dimensional
epsilon-tensor. Then, we are able to derive the following inequality,
\begin{equation}\label{iny}
  \vert {(By)_i\mathbf{e}^i}\vert\leq \sqrt{2} \vert \mathbf{B}\vert  \vert \mathbf{y}\vert.
\end{equation}
This is easily seen by using Cauchy-Schwarz and the inequality $\vert a\vert -\vert
b\vert\leq\vert a\vert+\vert b\vert$.
\begin{align*}
 \vert{(By)_i\mathbf{e}^i}\vert^2 &=(-B_{ij}y^jB^{is}y_s)\\
&=-\varepsilon_{ijk} B^k y^j  \varepsilon^{isr}B_r y_s\\
&= - \left( \delta_{j}^s\delta_{k}^r-\delta_{j}^r\delta_{k}^s\right)B_r B^k y^j   y_s\\
&=  \left( B_r y^r B_k y^k  -B_r B^r y^j   y_j\right)\\
& 
\leq 2 \vert \mathbf{B}\vert^2  \vert \mathbf{y}\vert^2 
\end{align*}
Thus, by using the adjoint action of $V(By)$ on the free Hamiltonian and for  $B_{ij}= \varepsilon_{ijk}B^k$   we have the following
inequality
\begin{align} \label{inecl4}
\vert {\langle \Psi, 
V(k)\alpha_{By}(H_0)\Phi\rangle}\vert& \nonumber 
\leq {\Vert \Psi\Vert} 
\left\Vert 
\left(H_{0}-(By)^s N_s-(By)_r(By)^s R_s^{\,\,r}\right)\Phi\right\Vert\\& 
\nonumber  
\leq \Vert \Psi\Vert \biggl(\underbrace{
\left\Vert  H_{0} \Phi \right\Vert}_{=:C_{2}}
+  2 \vert  \mathbf{y}
 \vert \underbrace{\frac{
 {\vert\mathbf{B}\vert}}{\sqrt{2}} 
\left\Vert  \mathbf{N} \Phi \right\Vert}_{=:C_{3}} + 
\vert  \mathbf{y}
 \vert ^2 \underbrace{ {2\vert\mathbf{B}\vert^2}   
\left\Vert\mathbf{R}\Phi \right\Vert}_{=:C_{4}}
\biggr)\nonumber \\
&\leq  {C_{B}\Vert \Psi\Vert } \left(1+\vert  \mathbf{y}
 \vert \right)^2.
\end{align}
A finite constant $C_{B}$ obeying the inequality  exists, since $ C_{2}$, $C_{3}$ and $C_{4}$ are finite for $\Phi \in \mathcal{E}$, where $\mathcal{E}$ is a dense set of vectors specified in Lemma \ref{cnt0}.  Therefore, the scalar product  is polynomially bounded to the second order in $y$, i.e. 

\begin{align*}\frac{
\vert\langle \Psi, (H_{0})_{B,\mathbf{Q}}\Phi\rangle\vert}{ C_{B}
{\Vert  \Psi\Vert}}&\leq
(2\pi)^{-3}   
\lim_{\varepsilon\rightarrow 0}
  \iint  \, d^{3}y \,  d^{3}k \, e^{-iy_{l}k^{l}}  \, \chi(\varepsilon y,\varepsilon
k) \left(1+\vert  \mathbf{y}
 \vert \right)^2 \\&
 =(2\pi)^{-3}
\lim_{\varepsilon_1\rightarrow 0}  \left(
\int d^3y \lim_{\varepsilon_2\rightarrow 0} 
\left(\int d^3k  e^{-ik_{r}y^r}
\chi_2(\varepsilon_2 k)\right)\,\chi_1(\varepsilon_1  y)\,
 \left(1+\vert  \mathbf{y}
 \vert \right)^2  \right)
\\
&=
\lim_{\varepsilon_1\rightarrow 0}  \left(
\int d^3y \,
\delta(\mathbf{y} )\,\chi(\varepsilon_1  y)\,
 \left(1+\vert  \mathbf{y}
 \vert \right)^2 \right) 
\\&= 1.
\end{align*}

Here we used the fact that the
oscillatory integral does not depend on the cut-off function  chosen. As in
\cite{RI}, we chose $\chi(\varepsilon k,\varepsilon y)= \chi_2(\varepsilon_2 k
)\chi_1(\varepsilon_1 y)$ 
with $\chi_{l}\in \mathscr{S}(\mathbb{R}^3\times\mathbb{R}^3)$ and $\chi_{l}(0,0)=1$, $l=1$, $2$,
and obtained the delta
distribution $\delta(\mathbf{y}-\mathbf{q})$ in the limit $\varepsilon_2 \rightarrow
0$, \cite[Section 7.8, Equation 7.8.5]{H}. Since the former inequality proves the convergence of the  oscillatory integral, we conclude that the deformation of the unbounded operator is well-defined.
\newline\newline Next, we turn to the actual result of the deformation. To simplify calculations it is easier to work in the spectral measure representation  (see Equation \ref{WC}). This can be done, since the two representations, one in terms of the spectral measure and the other as the limit of oscillatory integrals, are equal and we have proven that the deformation  is well-defined.  
\begin{align*}
(H_{0})_{B,\mathbf{Q}}\Phi
&=\int \, dE(y) \,\alpha_{By} \left(
H_{0}
\right)\Phi \nonumber \\&
=-\frac{1}{2m}\int \, dE(y) \, \left(   \left( 
P_{j}+i(By)_s [Q^s,P_j]    \right )   \left(
P^{j}+i(By)^r [Q_r,P^j]   \right)\right) \Phi\nonumber \\&
= -\frac{1}{2m}   \left( 
P_{j}+i(BQ  )_s [Q^s,P_j]      \right)   \left(
P^{j}+i(BQ  )^r  [Q_r,P^j]  \right)\Phi.
\end{align*}

\end{proof}$\,$\newline
The essential point of the proposition is that the deformation with the coordinate operator amounts to a
non-constant shift in the momentum. In physics this is usually referred to as \textbf{minimal
substitution}. Such a minimal substitution is in QM based on Galilei invariance and then
implemented accordingly by an external electromagnetic field (see \cite{LB}).  In our
approach we
obtain such a substitution by deformation. The connection between deformation and an
external electromagnetic field is explored in the next sections.
\newline\newline
For the next proposition we deform the momentum operator. Since the momentum operator is unbounded, we are as before obliged to show that deformation Formula (\ref{defcop}) is given as a
well-defined oscillatory integral. 
\begin{proposition}\label{lfps}Let 
$\mathbf{Q}(\mathbf{X})$ be a self-adjoint  operator of the form
\begin{align*}
 \mathbf{Q}(\mathbf{X}) =\mathbf{X}/ \vert\mathbf{X}\vert^{ n}, \qquad    n\in\mathbb{R},
\end{align*} and let $ \mathbf{P}_{ B,\mathbf{Q} }$ denote the deformed momentum operator (see Formula \ref{defcop}). Then, 
the scalar product $\langle \Psi, \mathbf{P}_{B,\mathbf{Q}}\Phi\rangle$ is bounded by a finite constant $C_{D}$ as follows,
\begin{align*}\vert
\langle \Psi, \mathbf{P}_{B,\mathbf{Q}}\Phi\rangle\vert&\leq
C_{D}\,\Vert \Psi\Vert, \qquad  \forall  \Psi\in \mathscr{H},\,\, \Phi \in \mathcal{E}\subseteq \mathscr{S}(\mathbb{R}^3).
\end{align*}  Therefore, the deformation of the unbounded momentum operator, given as an oscillatory
integral, is well-defined.  Moreover,  the explicit result of the deformation is given as
\begin{equation}\label{defp}
 {P}^j_{B,\mathbf{Q}}\Phi=\left(
 P^j+i(B  Q )^k [Q_k,P^j]  \right)
\Phi.
\end{equation}
 
\end{proposition}
\begin{proof} 
As in the proof of the former proposition we show that   $\vert{\langle \Psi, 
V(k)\alpha_{By}(\mathbf{P})\Phi\rangle}\vert$, is polynomially bounded.  To do so, we use the adjoint action of the unitary operator $V(By)$ on the momentum operator (see Equation (\ref{aaxp})) and   the Cauchy-Schwarz inequality, 
\begin{align*}
\vert{\langle \Psi, 
V(k)\alpha_{By}(\mathbf{P})\Phi\rangle}\vert& 
\leq \Vert \Psi\Vert  
\left\Vert 
\left( \mathbf{P} +i(By)^j [Q_j,\mathbf{P}]\right)\Phi\right\Vert\\&
\leq\Vert \Psi\Vert  
\left(\underbrace{
\left\Vert   \mathbf{P} \Phi \right\Vert}_{=:C_5}+  \vert  \mathbf{y}
 \vert \underbrace{
 {\sqrt{2}\vert\mathbf{B}\vert} 
\left\Vert  [\mathbf{Q},\mathbf{P}] \Phi \right\Vert}_{=:C_{6}}  \right) \\
&\leq C_{D}\Vert \Psi\Vert  
\left(1+\vert  \mathbf{y}
 \vert \right) 
,
\end{align*}
where  we used Inequality (\ref{iny}) and the fact that a finite constant $C_{D}$ obeying the inequality  exists, since $ C_{5}$ and $C_{6}$ are finite for $\Phi \in \mathcal{E}$ (see Lemma \ref{cnt0}).  Therefore,  the whole expression is polynomially bounded to  first  order in $y$, i.e. 
\begin{align*}\frac{
\vert\langle \Psi, \mathbf{P}_{B,\mathbf{Q}}\Phi\rangle\vert}{  C_{D}
{\Vert  \Psi\Vert}} &\leq 
(2\pi)^{-3}   
\lim_{\varepsilon\rightarrow 0}
  \iint  \, d^{3}y \,  d^{3}k \, e^{-iy_{l}k^{l}}  \, \chi(\varepsilon y,\varepsilon
k) \left(1+\vert  \mathbf{y}
 \vert \right)  \\&
 = (2\pi)^{-3}
\lim_{\varepsilon_1\rightarrow 0}  \left(
\int d^3y \lim_{\varepsilon_2\rightarrow 0} 
\left(\int d^3k  e^{-ik_{r}y^r}
\chi_2(\varepsilon_2 k)\right)\,\chi_1(\varepsilon_1  y)\,
 \left(1+\vert  \mathbf{y}
 \vert \right)   \right)
\\
&= 
\lim_{\varepsilon_1\rightarrow 0}  \left(
\int d^3y \,
\delta(\mathbf{y} )\,\chi(\varepsilon_1  y)\,
 \left(1+\vert  \mathbf{y}
 \vert \right)  \right) 
\\&= 1.
\end{align*}
As before, we argue that due to the convergence of the integral the deformation of the momentum operator for all $\Psi \in \mathscr{H}$ and $\Phi \in \mathcal{E}$   is well-defined. \newline\newline Next, we turn to the actual result of the deformation and again for simplicity we use the spectral measure for deformation,
\begin{align*}
 {P}^j_{B,\mathbf{Q}}\Psi
&=\int \, dE(y) \,\alpha_{By} \left(
 {P}^j
\right)\Psi \nonumber \\&
=\int \, dE(y) \,    \left( 
P^j+i(By)_s [Q^s,P^j]    \right )     \Psi\nonumber \\&
=   \left( 
P^j +i(BQ  )_s [Q^s,P^j]      \right)  \Psi.
\end{align*}
\end{proof}$\,$\newline
Since the deformed Hamiltonian could  be defined as the scalar product of the deformed momentum operators, we need to investigate the possible outcome. The investigation of the arbitrariness in the definition of the deformed free Hamiltonian is subject of the following theorem.

\begin{theorem}   
The scalar product of the deformed momentum vectors is equal to the deformed free Hamiltonian (see Equation \ref{defh}),
i.e. 
\begin{equation*}\label{d1}
 (H_{0})_{B,\mathbf{Q}}\Psi=-\frac{1}{2m}P_j^{B,\mathbf{Q}}P_{B,\mathbf{Q}}^j
\Psi, \qquad \Psi \in\mathcal{E}\subseteq \mathscr{S}(\mathbb{R}^3).
\end{equation*}
 
\end{theorem}

\begin{proof}
For the proof we calculate the Rieffel product, defined with the operator-valued
vector $\mathbf{Q}(\mathbf{X})$, of the deformed momentum vectors, i.e.
\begin{align*}
\left( P_{k}\times_{B,\mathbf{Q} }P_{j}\right)\Psi&= 
(2\pi)^{-3}
\lim_{\varepsilon\rightarrow 0}
\int\int d^{3}x\, d^{3}y \,\chi(\epsilon x,\epsilon y )\,e^{-ixy}\,  \alpha_{B x}( P_{k})\alpha_{y}(P_{j})\Psi\\&=
 \left(P_{k}P_{j}-iB_{ls}\partial_k Q^l  \partial_j Q^s\right)\Psi,
\end{align*}
where in the last lines we used the CCR and the fact that the only terms that do not vanish are those of
equal odd order in $x$ and $y$, (see proof of Lemma 5.3  in \cite{MUc}). Now by summing over all
components we obtain
 
\begin{align*}
\left( P_{k}\times_{B,\mathbf{Q} }P^{k}\right)\Psi&=  \left(P_{k}P^{k}-i\underbrace{B_{ls}\partial_k Q^l 
\partial^{k}
Q^s}_{=0}\right)\Psi= P_{k}P^{k}\Psi,
\end{align*}
  where we used the skew-symmetry of $B$ and commutativity of the coordinate operator. Thus, by
using the last equation and Lemma \ref{l2.1}  the following equality is given

\begin{align*}
 (H_{0})_{B,\mathbf{Q}}\Psi&=-\frac{1}{2m}\left( P_{k}P^{k}\right)_{B,\mathbf{Q}}\Psi
\\& =-\frac{1}{2m} \left( P_{k}\times_{B,\mathbf{Q} }P^{k}\right)_{B,\mathbf{Q}}\Psi
\\& =-\frac{1}{2m}P_k^{B,\mathbf{Q}}P_{B,\mathbf{Q}}^k
\Psi, \qquad \Psi \in \mathscr{S}(\mathbb{R}^3).
\end{align*}
\end{proof}$\,$\newline
This is an important result resolving the question of arbitrariness of the deformation. Moreover, it is a group
theoretical circumstance, since the deformation of a free Hamiltonian can   be understood as
the deformation of   generators of the central extended  Galilei (CEG) group (for CEG see for example
\cite{BAL}). The deformation with the coordinate operator leaves all generators of the group
invariant except for the momentum and the Hamiltonian. Since, the Hamiltonian is a function of the
momentum it follows from the former proposition that the deformation respects the structure of the
group. This fact is owed
to the deformed product. Also note that  the deformed momentum
operator \textbf{does not commute} along its components.
\newline\newline
As already mentioned in Section \ref{s2}, for some arguments we deform the coordinate operator by using the 
momentum operator. Before doing so, we show in the next proposition that the deformation formula is well-defined even though the coordinate operator is unbounded. 
\begin{proposition}\label{lfxs}
The scalar product $\langle \Psi, \mathbf{X}_{ \theta,\mathbf{P} }\Phi\rangle$ is bounded by a finite constant $C_{E}$ as follows,
\begin{align*}\vert
\langle \Psi, \mathbf{X}_{ \theta,\mathbf{P}  }\Phi\rangle\vert&\leq
C_{E}\,\Vert \Psi\Vert, \qquad  \forall  \Psi\in \mathscr{H},\,\, \Phi \in  \mathscr{S}(\mathbb{R}^3).
\end{align*}
Therefore, the deformation of the unbounded coordinate  operator, given as an oscillatory
integral, is well-defined.  Moreover,  the explicit result of the deformation is given as
\begin{equation}\label{defx}
 {X}^j_{\theta,\mathbf{P} }\Psi= \left( X^j-(\theta  P )^j\right) 
\Psi, \qquad \Psi \in  \mathscr{S}(\mathbb{R}^3).
\end{equation}
 
\end{proposition}

\begin{proof}
Similar to the proofs of the former propositions we show that the scalar product  $\vert{\langle \Psi, 
V(k)\alpha_{\theta y}(\mathbf{X})\Phi\rangle}\vert$ is polynomially bounded,
\begin{align*}
\vert{\langle \Psi, 
V(k)\alpha_{\theta y}(\mathbf{X})\Phi\rangle}\vert& 
\leq \Vert \Psi\Vert  
\left\Vert 
\left( \mathbf{X} +i(\theta y)^j [P_j,\mathbf{X}]\right)\Phi\right\Vert\\&
\leq\Vert \Psi\Vert  
\left(\underbrace{
\left\Vert   \mathbf{X} \Phi \right\Vert}_{=:C_7}+  \vert  \mathbf{y}
 \vert \underbrace{
 {\sqrt{2}\vert\mathbf{\theta}\vert} 
\left\Vert   \Phi \right\Vert}_{=:C_{8}}  \right) \\
&\leq C_{E}\Vert \Psi\Vert  
\left(1+\vert  \mathbf{y}
 \vert \right) 
,
\end{align*}
where in the last lines we used Inequality (\ref{iny}), and the fact that a finite constant $C_{E}$ obeying the inequality  exists, since $ C_{7}$ and $C_{8}$ are finite for $\Phi \in \mathscr{S}(\mathbb{R}^3)$.  Therefore,  the whole expression is polynomially bounded to the first  order in $y$, i.e. 
\begin{align*} {
\vert\langle \Psi, \mathbf{X}_{ \theta,\mathbf{P}  }\Phi\rangle\vert} &\leq 
 C_{E}
{\Vert  \Psi\Vert},
\end{align*}
where we used the same arguments made in Proposition \ref{lfps}.
  \newline\newline Next, we turn to the result of the deformation and again for simplicity we use the spectral measure  for the deformation,
 
\begin{align*}
 {X}^j_{\theta,\mathbf{P}   }\Psi
&=\int \, dE(y) \,\alpha_{\theta y} \left(
 {X}^j
\right)\Psi \nonumber \\&
=\int \, dE(y) \,    \left( 
X^{j}+i(\theta y)_s [P^s,X^j]    \right )     \Psi\nonumber \\&
=   \left( 
X^{j}-(\theta P  )^{j}    \right)  \Psi.
\end{align*}
After this rather more technical part we turn in the next section to the physical implications of the deformation technique.

 \end{proof}
 
  \section{Physical models from deformation}\label{s4}

One of the most important aspects of the interplay between mathematics and physics lies in
the \textbf{physical dimensionality} of the physical constants. The main motivation of this work is the search for the physical meaning of the deformation
parameter. Quantum mechanical deformations give us a variety of interesting answers and they are presented in this section.
\subsection{Landau quantization} 
An example of a dynamical system interacting with a magnetic field in a
quantum mechanical setting, is given by the Landau effect. It is also an important example of the
appearance of quantum space in a physical context. The Landau effect describes the
dynamics of a system of non-relativistic electrons \textbf{confined to a plane}, for example the
$y-z$ plane ($\vec{A}=(0,y,z)$), in the presence of a homogeneous magnetic
field  
$\vec{B}=B(1,0,0)$. In the \textbf{symmetric gauge} the Hamiltonian of the Landau effect is
given by, \cite[Equation 9.2.1]{EZ}
\begin{equation*}
 H_{L}=-\frac{1}{2m}\left(P_{i}+eA_{i}\right)\left(P^{i}+eA^{i}\right),
\end{equation*}
where the gauge field is given as 
\begin{equation}\label{gf1}
A_{i}=-\frac{1}{2}\varepsilon_{ijk}{B}^{k}X^{j}.
\end{equation} 
Next, we show that the deformed Hamiltonian $(H_{0})_{B,\mathbf{X}}$ reproduces the Landau
model after setting the parameters  of the deformation matrix  
equal to a constant with physical dimension. This is the result of the following lemma. 
\newline  
\begin{lemma}\label{llq}
Let the \textbf{deformation matrix}   $B_{ij} $ be given as,
\begin{equation*}
B_{ij}=-(e/2)\,\varepsilon_{ijk}B^{k},
\end{equation*}
 where $B^k$ characterizes a constant  
homogeneous magnetic field and  $e$ is the electric charge. Then, the deformed free
Hamiltonian
$(H_{0})_{B,\mathbf{X}}$ becomes the Hamiltonian  $H_{L}$ of the Landau  problem, i.e. 
\begin{equation*}
 (H_{0})_{B,\mathbf{X}}\Psi= H_{L}\Psi, \qquad \Psi \in   \mathcal{E}.
\end{equation*}
\end{lemma}
\begin{proof}
For the proof we consider the free deformed Hamiltonian $(H_{0})_{B,\mathbf{X}}$,  given in Equation (\ref{defh}), with $Q_{j}=X_{j}$ 
\begin{align} \label{defhx}\nonumber
(H_{0})_{B,\mathbf{X}}\Psi&=
-\frac{1}{2m}\left
(P_j+i(B  X )^k [X_k,P_j]  \right)\left(P^j+i(B  Q )^r [X_r,P^j]  \right)\Psi\\&=
-\frac{1}{2m}
(P_j+B_{jk}X^k)(P^j+B^{jr}X_r)\Psi, \qquad \Psi \in   \mathcal{E},
\end{align}
where in the last line we used the CCR. 
By setting the deformation matrix equal to $B_{ij}=-(e/2)\,\varepsilon_{ijk}B^{k}$, where $B^{k}$ is a homogeneous magnetic
field in the $x$-direction ($B^{k}=B(1,0,0)$), we obtain the Landau quantization. 
\end{proof}$\,$\newline
This is an interesting result. We started with the free Hamiltonian and deformed it with warped
convolutions
using the coordinate operator. By simply taking the deformation parameters of the matrix $B_{ij}$ to be equal to certain
 physical quantities   we
obtain the Landau problem. Therefore, the quantization with the coordinate operator is
physically of great importance. Note that our model is formulated in a general manner, and just
for the specific choice of the deformation parameters  we obtained the Landau
effect. 
\newline\newline A remark is in order about the current result. It is well-known  that the noncommutative coordinates of the Landau quantization can be generated by minimally shifting the ordinary coordinate operator by a skew-symmetric matrix times the momentum operator. This rather ad hoc  but remarkably insightful  result is well-known as the Bopp-shift. In the context of deformation theory, we were able to give a systematic derivation of the Landau quantization, rather than postulating ad hoc a substitution. This derivation can be further applied to a variety of quantum mechanical effects involving gauge fields.
\subsection{Zeemaneffect }
The Hamiltonian of the hydrogen atom is
given as follows, \cite[Equation 4.1.1]{Th}

 \begin{equation*}\label{ha}
H^{A}{}= -\frac{P_jP^j}{2m}+\frac{e^2}{\vert\mathbf{X}\vert}.
\end{equation*}
By solving the stationary Schr\"odinger equation $H^{A}\psi=E\psi$ one obtains the energy
spectrum of a hydrogen atom, the so called \textbf{Balmer series}, \cite{St}. In the presence of
a constant magnetic field,
an interesting physical effect occurs to the spectral lines of the hydrogen atom. The
spectral lines split into further spectral lines depending on the presence of a
homogeneous magnetic field ${B}_{k}$. This phenomenon is called the  \textbf{Zeemaneffect} and the
Hamiltonian of this effect is given as follows, \cite[Equation 4.2.1]{Th}
 \begin{equation}\label{hab}
H^{AZ}=
-\frac{1}{2m}{(P_j-(e/2)\,\,\varepsilon_{jik}B^{k}X^i)(P^j-(e/2)\,\,\varepsilon^{jnl}B_{l}X_n)}+\frac{e^2}{
\vert\mathbf{X}\vert}.
\end{equation} 
We recognized in the last section that the deformation with the coordinate operator
induces a gauge field. Due to this lesson we preform a deformation on the Hamiltonian of
the hydrogen atom to obtain the Hamiltonian of the Zeemanneffect.  
\newline  
\begin{lemma}\label{cnt2}
Let the \textbf{deformation matrix}   $B_{ij} $ be given as,
\begin{equation*}
B_{ij}=-(e/2)\,\varepsilon_{ijk}B^{k},
\end{equation*}
 where $B^k$ characterizes a constant  
homogeneous magnetic field.  Then, the deformed Hamiltonian
of the hydrogen atom, denoted by $(H^{A})_{B,\mathbf{X}}$, becomes the
Hamiltonian of the Zeemaneffect $H ^{AZ}$, i.e.
\begin{equation*}
(H^{A})_{B,\mathbf{X}}\Psi= H ^{AZ}\Psi, \qquad \Psi \in   \mathcal{E}.
\end{equation*}
 
\end{lemma}
\begin{proof}
Due to the fact that the coordinate operator commutes with itself the only part of the
Hamiltonian $H ^{A}$ which is affected is the free part and therefore we obtain 
 \begin{equation}
(H^{A})_{B,\mathbf{X}}\Psi=\biggl(
-\frac{1}{2m}
(P_j+B_{jk}X^k)(P^j+B^{jr}X_r)+\frac{e^2}{\vert\mathbf{X}\vert}\biggr)\Psi,
\end{equation} 
the Hamiltonian of the Zeemaneffect for a homogeneous magnetic field in the
$x$-direction, i.e. $(H^{A})_{B,\mathbf{X}}= H ^{AZ}$. 
\end{proof}$\,$\newline
As in the case of Landau quantization, the deformation parameter plays the role of the magnetic
field which leads to
this wonderful physical effect.  
\subsection{Aharonov-Bohm effect}
In the last sections we recognized the consequence of a deformation with the coordinate
operator. Warped convolutions with the coordinate operator induce a \textbf{gauge field}. Now
since we work in a quantum mechanical setting we want to  
reproduce other physical effects where magnetic fields play a significant role. One of
the most striking ones is the  \textbf{Aharanov-Bohm (AB) effect}. It takes place in a system in
which the gauge field   influences the dynamics of a charged
particle even in regions where the magnetic field  vanishes, \cite{Be,EZ}. 
 The gauge field  of the magnetic AB
effect, for a homogeneous magnetic field in $x$-direction, takes the following form 
\begin{equation}\label{cntab}
 A_{i}=\frac{\phi_{M}}{2\pi(X_{2}^2+X_{3}^2)} \varepsilon_{ijk}e^{k} X^{j},
\end{equation}
where $\phi_{M}$ is the  \textbf{magnetic flux} and $e^{k}$ is the unit vector in
$x$-direction. Moreover, from quantum  mechanical considerations it follows that the
\textbf{interference  pattern} is the same for two values of fluxes $\phi_{ 1}$ and $\phi_{ 2}$
if only if 
\begin{equation}\label{mfdp}
 e(\phi_{1}-\phi_{2 })=2\pi n, \qquad n \in \mathbb{Z}.
\end{equation} In this section we take the free Hamiltonian  and deform it with a vector-valued function of the coordinate operator. As before, after  
setting  the   deformation parameter equal to a physical
constant, namely that of a magnetic flux, we obtain the AB effect.

\begin{proposition}
Let the \textbf{deformation matrix}  $B_{ij}$ and the operator $Q_j(\mathbf{X})$  be given as
\begin{equation}\label{opab}
B_{ij}=- \frac{e\,\phi_{M}}{2\pi} \varepsilon_{ijk}e^{k} ,\qquad Q_j(\mathbf{X}):={X_j}/{( -\sum\limits_ {s=2}^3X_sX^s)^{1/2}},
\end{equation}  where $\phi_{M}$ characterizes the magnetic flux.
Then, the deformed Hamiltonian $(H_0)_{B,\mathbf{Q}}$, is equal to Hamiltonian of the Aharonov-Bohm, i.e.  
\begin{align*} 
(H_{0})_{B,\mathbf{Q}}\Psi
&=  \frac{1}{2m}  \left( 
\mathbf{P} -e \mathbf{A}   \right)^2 
 \Psi,
\end{align*} 
where $\mathbf{A}$  is the gauge field of the Aharanov-Bohm effect (see Equation \ref{cntab}).
Furthermore, if the deformation
parameters of the matrices $B_1$ and $B_2$ fulfill
Equation (\ref{mfdp}), the physical systems described by the Hamiltonians 
 $H_{B_1, F }$ and $H_{B_2, F }$   have the same
interference pattern.
\end{proposition}
\begin{proof}  For the deformation of
$H_{0}$ we use Proposition \ref{wcfh}, with $Q_j(\mathbf{X})$ as given in Equation (\ref{opab}),
\begin{align*}\label{defhab}
(H_{0})_{B,\mathbf{Q}}\Psi
&=-\frac{1}{2m}\left
(P_j+i(B  Q )^k [Q_k,P_j]  \right)\left(P^j+i(B  Q )^r [Q_r,P^j]  \right) \Psi
\nonumber \\&
=  -\frac{1}{2m}  \left( 
P_{j}+(BX)_j/ (-\sum_{s=2}^{3}X_sX^s)   \right) \left( 
P^{j}+(BX)^j/ (-\sum_{r=2}^{3}X_rX^r)   \right) 
 \Psi, 
\end{align*}
where in the last lines we used the skew-symmetry of $B$ and commutator relation \ref{hs2}.
 Thus, by setting the deformation matrix $B_{ij}= - ({e\,\phi_{M}}/{2\pi}) \varepsilon_{ijk}e^{k}$,
the gauge field $A_i(\mathbf{x})$   induced by deformation  is  the gauge field of the AB effect for a homogeneous magnetic field in
$x$-direction.
  \end{proof}$\,$\newline
This is an interesting result. We were able to induce the AB-gauge field by deforming the free
Hamiltonian with a vector-valued function of the  coordinate operator.  In this
case the deformation parameter corresponds to the magnetic flux rather, as in the previous cases,
to the magnetic field.
\newline\newline
There are two ways to interpret these results. The first one lies in understanding deformation, in
the case of QM, as the rightful minimal substitution. Thus the  procedure sheds new light on
quantum mechanical effects involving
magnetic fields. The fields can be understood as the outcome of a deformation with vector-valued functions of
the coordinate operator.  The other way of understanding the result is the following. The coupling
of an external magnetic
field in QM is well understood and studied for various physical applications and models.
Deformation on the other hand is a mathematical tool, rather than a procedure that generates
physical effects. Hence, in these examples deformation of  a QM system can be understood as the
coupling of an external field. Thus, if the deformation goes hand in hand with Moyal-type
spaces one sees in these examples that Moyal spaces correspond to ordinary spaces in
the presence of an external field. By having this observation in mind it does not seem far fetched
that certain deformations of spacetime correspond to gravitation. Let us describe in the next sections how Moyal-Weyl spaces arise in this context.

\subsection{Physical Moyal-Weyl plane  }\label{pmwfd}
 To describe the circular motion of an electron in the lowest Landau level  we define
the so called  guiding
center coordinates   $\mathbf{Q}$,
\cite{EZ, SZ}  
\begin{equation*}
 Q_{i}:= X_{i}+\frac{1}{2}(B^{-1})_{ik}P^{k},
\end{equation*}
 with matrix $B_{ij}=-(e/2)\,\varepsilon_{ijk}B^{k}$. Note that the inverse corresponds to the non-degenerate sub-matrix of $B_{ij}$.
By using the CCR it becomes apparent that the guiding center coordinates span a  three
dimensional Moyal-Weyl plane, i.e. 
\begin{equation}\label{mw}
 [Q_{i},Q_{j}]= i(B^{-1})_{ij}.
\end{equation}
  Thus, the Landaueffect is an example of a physical noncommutative space.
Now can we generate these noncommuting coordinates by the deformation procedure warped convolutions?  Yes we can!  
  
\begin{lemma}\label{lmw}
The deformed coordinate operators $ X^j_{\theta,
\mathbf{P}}$  given as  (see Equation \ref{defx})
\begin{equation}
 X^j_{\theta,
\mathbf{P}} =
X^j-\theta^{jr}P_r,
\end{equation}
satisfy the commutation
relations of the Moyal-Weyl plane $\mathbb{R}^3_{-2\theta}$,
\begin{equation}\label{mwx}
 [X^{i}_{\theta,
\mathbf{P}},X^{j}_{\theta,
\mathbf{P}}]=-2i\theta^{ij}.
\end{equation}
Moreover, let $-2\theta^{ij}$ be $(B^{-1})^{ij}$ then the deformed coordinate operators
$X^{i}_{\theta,
\mathbf{P}}$ are equal to the guiding center coordinates given in Equation (\ref{mw}).
\end{lemma}
\begin{proof} 
 
The commutator of the deformed coordinate operator is calculated by using the
canonical commutation relations  and the skew-symmetry of the deformation
matrix
$\theta_{jk}$.
\begin{align*}
 [X^j_{\theta,
\mathbf{P}},X^k_{\theta,
\mathbf{P}}]&=
 [X^j-\theta^{jr}P_r,X^k-\theta^{kl}P_l]\\&=
-\theta^{kl}[X^j,P_l]+\theta^{jl}[X^k,P_l]\\&
=-2i\theta^{jk}.
\end{align*} 
\end{proof}$\,$\newline
Lemma \ref{lmw} gives a well defined path to obtain an effective quantum plane by the
deformation using warped convolutions. As we showed, the lemma follows from well
understood physical models and ideas, which are in circulation in condensed matter field
theory, for quite some time. In the example of the Landau problem one defines guiding center
coordinates,
which satisfy the commutator relations of the Moyal-Weyl Plane. The reader is cautioned
to notice that the effective quantum plane obtained by the Landau problem is not merely
an abstract construct but has the precise meaning, that the space coordinates can not be
measured simultaneously. A more precise mathematical way to obtain this Moyal-Weyl plane
is introduced in this work. We obtain the Landau problem by deforming the Hamiltonian of
a free non-relativistic particle with the coordinate operator and
by setting the deformation parameter equal to a magnetic field. Furthermore, we show that the noncommuting
coordinates referred to as the guiding coordinates are obtained by deforming the
coordinate operator, using the momentum operator. In our opinion, this method can be
further used in the quantum field theoretical (QFT) approach to define an effective quantum
plane.  
\subsection{Gravitomagnetism in QM}
The emergence of gravitomagnetism in QM   from deformation theory is one of the centerpieces of this work. Before we prove the emergence of these effects let us introduce some
 basic notations. We consider  slowly varying weak gravitational fields with   energy momentum
tensor of ordinary matter (dust-like). In this description the
metric can be written as
\begin{equation*}
 g_{\mu\nu}=\eta_{\mu\nu}+h_{\mu\nu},
\end{equation*}
where $h$ is the small perturbation from the flat spacetime and the
energy momentum tensor can be written as 
\begin{equation*}
 T_{\mu\nu}=\rho u_{\mu }u_{ \nu},
\end{equation*}
where $u$ is the $4$-velocity and $\rho$ the scalar density.  For slowly varying fields, the linearized Einstein field equations
   can be
described by Maxwell-like field equations given by
\begin{equation*}
 \Delta \phi=4\pi G \rho, \qquad  \Delta h^j=-16\pi G \rho  v^j, \qquad
4\dot{\phi}-\nabla\cdot  h^j=0,
\end{equation*}
with the definitions of the potentials
\begin{equation*}
\phi:=h_{00}/2,\qquad h^{j}:=h^{0j},
\end{equation*}
where we  used the Lorentz condition and the fact that the fields considered are
slowly varying, i.e. $\ddot{\phi}$, $ \dot{h}_k$ and $\ddot{h}_k $  can be neglected, (see for
example \cite{AP}, 
\cite{WS}). Analogously to the electromagnetic case the gravitoelectric field $\mathbf{g}$ and the
gravitomagnetic  field $\mathbf{\Omega}$ are both defined by the potentials as

\begin{equation*}
\mathbf{g}=-\boldsymbol{\nabla} \phi, \qquad  \mathbf{\Omega}=\boldsymbol{\nabla}  \times \mathbf{h}.
\end{equation*}
There are a few important examples that can be considered in the context  gravitomagnetism. One of them is 
example of  the vector potential $\mathbf{h}$ inside a hollow spinning sphere  with radius $r_{hs}$ and spin $\boldsymbol{\omega}$ that is
given by, \cite[Eq. 9.4.35]{WS}
\begin{equation}\label{gmp}
\mathbf{h}(\mathbf{x})= \mathbf{x}\times \mathbf{\Omega},
\end{equation}  
where ${\Omega}= {2MG}/{r_{hs}}$ is the  constant gravitomagnetic field  inside the hollow  sphere.\newline\newline
Now in \cite{AP} the authors derived the non-relativistic Schr\"odinger equation  for a
particle that is minimally coupled to an external electromagnetic and gravitoelectromagnetic
field. The equation is given by \cite[Eq. 5.1]{AP},

\begin{equation}\label{cnt5}
H_{GEM}{\Psi}=-\frac{1}{2m}\left(P-eA\right)_{i}\left(P-eA\right)^{i}{\Psi}-
h^{i}\left(P-eA\right)_{i}{\Psi},
\end{equation}
where we set the potential $V,\phi$ equal to zero and neglected the term of second derivative in
$\Psi$. This can be done since the term is just a relativistic correction which for slowly
moving bodies can be neglected.  \newline\newline
By using the former definitions  and results we are  able to
reproduce the case of a constant gravitomagnetic field by deformation.
 
\begin{lemma}\label{llq1}
Let the deformation  matrix $B_{ij}$   be given as 
\begin{equation}
 B_{ij}=m\,\varepsilon_{ijk}\,\Omega^{k} ,
\end{equation}
where  ${\Omega}^{k}=(2GM/r_{hs})\, {\omega}^{k}$ is a constant gravitomagnetic field for a hollow
spinning
sphere. Then, the deformed free Hamiltonian $(H_{0})_{B,\mathbf{X}}$,  becomes the Hamiltonian
 of a quantum mechanical particle minimally coupled to a constant gravitomagnetic field, i.e.
\begin{align*} 
(H_{0})_{B,\mathbf{X}}\Psi(\mathbf{x})&
=-\frac{1}{2m}
(P_j+m\,h_{j})(P^j+m\,h^{j})\Psi\\&= H_{0}\Psi(\mathbf{x})-h^{j}P_{j}\Psi(\mathbf{x})+\mathcal{O}(h^2),\qquad \Psi\in\mathcal{E},
\end{align*}
where the vector  $h_{j}=  \varepsilon_{jkl} \,x^{k}\,{\Omega}^{l}$ represents the gravitomagnetic vector potential for a hollow
spinning sphere (see Equation (\ref{gmp})).
\end{lemma}
\begin{proof}
The free deformed Hamiltonian $(H_{0})_{B,\mathbf{X}}$  is given by
\begin{align*} 
(H_{0})_{B,\mathbf{X}}\Psi&=
-\frac{1}{2m}
(P_j+B_{jk}X^k)(P^j+B^{jr}X_r)\Psi\\&
=-\frac{1}{2m}
(P_j+m\,h_{j})(P^j+m\,h^{j})\Psi\\&= H_{0}\Psi(\mathbf{x})-h^{j}P_{j}\Psi(\mathbf{x})+\mathcal{O}(h^2),\qquad \Psi\in\mathcal{E},
\end{align*} where in the last line we set the deformation  matrix
$B_{ij}=  m\,\varepsilon_{ijk}(2GM/r_{hs})\,{\omega}^{k}$ and we neglected second order terms  since we work in the
linear approximation. 
\end{proof}$\,$\newline
This an important result, since this means that we obtain gravitational effects from a
well-defined deformation procedure by simply adjusting the deformation constants accordingly. Thus,  gravitomagnetism can be understood as the outcome of a deformation
procedure. Moreover, the physical constant  used as deformation parameter in the gravitomagnetic case is the gravitational constant $G$. Since by setting the gravitational constant to zero, i.e. neglecting gravitational effects, the deformed Hamiltonian describing gravitomagnetic effects becomes the free Hamiltonian.   \newline\newline
Next we use Proposition \ref{wc2} of the deformation technique to obtain the electromagnetic
and gravitomagnetic coupling.

\begin{proposition}
Let the deformation  matrix $B^{1}_{ij}$   be given as 
\begin{equation*}
 B^{1}_{ij}=m\,\varepsilon_{ijk}\,\Omega^{k},
\end{equation*}
where  $\Omega^{k}=(2GM/r_{hs}) {\omega}^{k}$ is the constant gravitomagnetic field for a hollow
spinning
sphere and let the deformation  matrix $B^{2}_{ij}$  be given as 
\begin{equation*}
 B^{2}_{ij}=-(e/2)\,\varepsilon_{ijk}B^{k},
\end{equation*}
where $B^k$ is a homogeneous magnetic field.\newline\newline
Then, the deformed free Hamiltonian $\left((H_0)_{B^{1},\mathbf{X}}\right)_{B^{2},\mathbf{X}}$ becomes the Hamiltonian
$H_{GEM}$ (see Equation \ref{cnt5}) of a quantum mechanical particle minimally coupled to  
a constant external magnetic and gravitomagnetic 
field, i.e.
\begin{equation*}
(H_0)_{B^{1}+B^{2},\mathbf{X}}\Psi=H_{GEM}\Psi,\qquad \Psi\in\mathcal{E.} 
\end{equation*}

\end{proposition}
 \begin{proof}
First of all by the virtue of Proposition \ref{wc2} the deformed Hamiltonian satisfies 
\begin{equation*} 
\left((H_0)_{B^{1},\mathbf{X}}\right)_{B^{2},\mathbf{X}}\Psi =(H_0)_{B^{1}+B^{2},\mathbf{X}}\Psi.
\end{equation*}
 Next we consider the free deformed Hamiltonian $(H_0)_{B^{1}+B^{2},\mathbf{X}}$, (see Equation
(\ref{defhx})).
\begin{align*} 
(H_0)_{B^{1}+B^{2},\mathbf{X}}\Psi&=
-\frac{1}{2m}
\left(P+\left((B^{1}+B^{2})X\right)\right)_j\left(P+\left((B^{1}+B^{2})X\right)\right)^j \Psi\\&=
-\frac{1}{2m}
(P_j-eA_{j}+
mh_{j}) (P^j-eA^{j}+
mh^{j}) \Psi
\\&=-\frac{1}{2m}\left(P-eA\right)_{j}\left(P-eA\right)^{j}{\Psi}-h_{j}\left(P-eA 
\right)^{j}\Psi+\mathcal{O}(h^2),
\end{align*}
where we set the deformation  matrix
$B^1_{ij}= m\varepsilon_{ijk}(2GM/r_{hs}) {\omega}^{k}$ and
$B^{2}_{ij}= -(e/2)\varepsilon_{ijk}B^{k}$ and $h_{j}$ is the gravitomagnetic vector potential given in Equation (\ref{gmp}) for a hollow
spinning sphere  and $-A^{j}$  the  magnetic vector potential given in the Landau quantization,
(see Equation (\ref{gf1})).
\end{proof}$\,$\newline
From this result it becomes clear that in a quantum mechanical setting one can obtain electromagnetic
and gravitomagnetic effects by a deformation procedure. In the framework of deformation these
effects simply correspond to certain deformation parameters that in turn are given by physical
constants. One should also note that we obtained the Hamiltonian of a quantum mechanical system
that is coupled to an  external magnetic and gravitomagnetic field by 
deformation, rather than by advanced calculations and considerations as done in \cite{AP}.
\subsection{Lense-Thirring Precession}
Another important gravitomagnetic effect is known under the name of Lense–Thirring precession.
The effect is a general-relativistic correction to the precession of a gyroscope outside a massive
stationary spinning sphere. The vector potential for such gravitomagnetic field is given as  
\begin{equation}\label{vplt}
 \mathbf{h}=-(2GI/r^3)\mathbf{x}\times\boldsymbol{\omega},
\end{equation}
where $I$ is the moment of inertia of the sphere and $r=\vert\mathbf{x} \vert$ the radius. 
\newline\newline
As for the constant gravitomagnetic field, we are also able to produce the vector potential of the
Lense–Thirring effect.  
  
\begin{proposition}
Let the deformation matrix $B_{ij}$  and the operator $Q_j$ be given as
\begin{equation}\label{lteo}
 B_{ij}=m\,\varepsilon_{ijk}\,\Omega^{k}, \qquad Q_j(\mathbf{X})={X_j}/{\vert\mathbf{X}\vert^{3/2}}.
\end{equation}
where $\Omega^{k}= (2GI)\, {\omega}^{k}$ and $I$ is the moment of inertia of a spinning sphere.  
Then, the deformed free Hamiltonian $(H_{0})_{B,\mathbf{Q}}$,   becomes the Hamiltonian of a quantum mechanical particle minimally coupled to the  gravitomagnetic field of the
Lense-Thirring effect, i.e.
\begin{align*} 
(H_{0})_{B,\mathbf{Q}}\Psi &=-\frac{1}{2m}  \left( 
P_{j}+(BX)_j/ {\vert\mathbf{X}\vert^{3}}  \right)  \left( 
P^{j}+(BX)^j/ {\vert\mathbf{X}\vert^{3}}  \right)  \Psi ,\qquad \Psi\in\mathcal{E } \\&
=H_{0}\Psi-h_jP^j\Psi+\mathcal{O}(h^2),
\end{align*}
 where the vector potential induced   by deformation is the gauge field of the Lense-Thirring effect, i.e. $h_j=m\,\varepsilon_{jkl}(2GI)\, {\omega}^{l}X^k/ {\vert\mathbf{X}\vert^{3}}$.

\end{proposition}
 
\begin{proof} 
To prove this proposition we use the spectral measure representation.  The deformation of
$H_{0}$ is then given as follows,
\begin{align*}\label{defhab4}
(H_{0})_{B,\mathbf{Q}} \Psi \nonumber &
= -\frac{1}{2m}  ( 
P_{j}+(BQ )_s  [Q^s, P_j]   )  (
P^{j}+(BQ )^r [Q_r, P^j]  )\Psi\nonumber \\
&=-\frac{1}{2m}  \left( 
P_{j}+(BX)_j/ {\vert\mathbf{X}\vert^{3}}  \right)  \left( 
P^{j}+(BX)^j/ {\vert\mathbf{X}\vert^{3}}  \right) \Psi\\& 
=-\frac{1}{2m}\biggl(P_j+m\, h_j\biggr) \biggl(P^j+m \,h^j\biggr)  \Psi
\\&
=H_{0}\Psi-h_jP^j\Psi+\mathcal{O}(h^2),
\end{align*}
where in the last lines we used the skew-symmetry of $B$ and the commutator relation $\ref{hs2}$.
  \end{proof}$\,$\newline
Next, we use   the deformation technique to obtain the electromagnetic
and gravitomagnetic coupling in the case of the Lense-Thirring effect. The effects emerge by a
double deformation where once we use the coordinate operator and after that the operator-valued
vector $Q_j(\textbf{X})$. Note that the order of the deformation is irrelevant, since the two operators commute.

\begin{remark}
The proof that the deformation with two different operators is well-defined, is equivalent to proving Proposition \ref{wcfh}, where one replaces the free Hamiltonian in Inequality (\ref{inecl4}) with  $(H_0)_{{B^{2} ,\mathbf{X}}}$. It then follows that for $  \Phi \in \mathscr{S}(\mathbb{R}^3) $, the expression $\Vert (H_0)_{{B^{2} ,\mathbf{X}}} \Phi \Vert $ is finite.
 \end{remark}

\begin{proposition}\label{cnt3}

Let the deformation  matrix $B^{1}_{ij}$  be given as 
\begin{equation}\label{dpg1}
 B^{1}_{ij}=m\,\varepsilon_{ijk}\Omega^{k},
\end{equation}
where $\Omega^{k}= (2GI)\, {\omega}^{k}$  and let the deformation  matrix $B^{2}_{ij}$  be given as 
\begin{equation}\label{dpg2}
 B^{2}_{ij}= -(e/2)\,\varepsilon_{ijk}B^{k},
\end{equation}
where $B^k$ is a homogeneous magnetic field. Moreover, let the operator $Q_j(\textbf{X})$ be given by
\begin{equation*} 
Q_j(\mathbf{X})={X_j}/{\vert\mathbf{X}\vert^{3/2}}.
\end{equation*}
Then, the deformed free Hamiltonian $((H_{0})_{B^{2} ,\mathbf{X}})_{B^{1}, \mathbf{Q}}$ becomes the
Hamiltonian
$H_{GEM}$ of a quantum mechanical particle minimally coupled to  
a constant external magnetic and the gravitomagnetic 
field of the Lense-Thirring effect, i.e.
\begin{equation*}
((H_0)_{{B^{2} ,\mathbf{X}}})_{B^{1}, \mathbf{Q} }\Psi=H_{GEM}\Psi,\qquad \Psi\in\mathcal{E.} 
\end{equation*}

\end{proposition}

\begin{proof}
 The free deformed Hamiltonian $(H_0)_{{B^{2} ,\mathbf{X}}}$ is given by, (see Equation
(\ref{defhx})).
\begin{equation*} 
(H_0)_{{B^{2} ,\mathbf{X}}}\Psi=
-\frac{1}{2m}
\left(P+\left(  B^{2} X\right)\right)_j\left(P+\left( B^{2} X\right)\right)^j\Psi .
\end{equation*}
Due to the commutativity of the coordinate operators, deformations with  $Q_j(
\mathbf{X})$ do not influence the gauge field  $(B^{2}  X )$ and vice versa, i.e.  $((H_{0})_{B^{2}
,\mathbf{X}})_{B^{1}, \mathbf{Q}}=((H_{0})_{B^{1}, \mathbf{Q}})_{B^{2} ,\mathbf{X} }$. Thus,  after choosing the deformation parameters as stated in Equations (\ref{dpg1}) and  (\ref{dpg2}) we
obtain for the deformed free Hamiltonian,
\begin{align*}
 ((H_{0})_{B^{1}, \mathbf{Q}})_{B^{2} ,\mathbf{X}}\Psi& 
=-\frac{1}{2m}
\left(P   -eA  + m \,h   \right)_j\left(P   -eA  + m\, h   \right)^j\Psi
\\&=-\frac{1}{2m}\left(P-eA\right)_{j}\left(P-eA\right)^{j}{\Psi}-h_{j}\left(P-eA 
\right)^{j}\Psi+\mathcal{O}(h^2),
\end{align*}
  where $h_{j}$ is the gravitomagnetic gauge field of the Lense-Thirring effect (see Equation
(\ref{vplt})) and $-A^{j}$ the magnetic vector potential given in the Landau quantization,
(see Equation (\ref{gf1})).
\end{proof}
 
 \subsection{Gravitomagnetic Zeemaneffect} 
Similar  to the magnetic case, where  the Zeemaneffect emerged by deforming the hydrogen Hamiltonian with the same deformation matrix used in the Landau quantization, we precede in the gravitomagnetic case. Thus, we are able to predict a gravitomagnetic Zeemaneffect by deforming the hydrogen atom and using the constant gravitomagnetic deformation matrix.

\begin{lemma}Let the \textbf{deformation matrix}   $B_{ij} $ be given as,
\begin{equation*}
B^1_{ij}=m\,\varepsilon_{ijk}\Omega^{k},
\end{equation*}
where  $\Omega^{k}=(2GM/r_{hs}) {\omega}^{k}$ is the constant gravitomagnetic field for a hollow spinning
sphere.   Then, the deformed Hamiltonian
of the hydrogen atom, denoted by $(H^{A})_{B^1,\mathbf{X}}$, becomes the
Hamiltonian of the gravitomagnetic Zeemaneffect, i.e.
\begin{align*}
(H^{A})_{B^1,\mathbf{X}}\Psi&= -\frac{1}{2m}{(P_j+m\,\varepsilon_{jkl}\,\Omega^{l}X^k)(P^j+m\,\varepsilon^{jrs}\,\Omega_{s}X_r)}\Psi +\frac{e^2}{
\vert\mathbf{X}\vert}\Psi, \qquad \Psi \in   \mathcal{E}\\&= -\frac{1}{2m}{(P_j+m\,h_{j})(P^j+m\,h^{j})}\Psi +\frac{e^2}{
\vert\mathbf{X}\vert}\Psi
\\&
=H_{0}\Psi-h_jP^j\Psi+\frac{e^2}{
\vert\mathbf{X}\vert}\Psi+\mathcal{O}(h^2)
.
\end{align*}
 \begin{proof} The only difference to the proof of Lemma \ref{cnt2} consists in the choice of the deformation matrix, i.e. the proof is equivalent. 
\end{proof}
\end{lemma}$\,$\newline
Analogously to the magnetic case, the presence of
a constant gravitomagnetic field will split the spectral lines of the hydrogen atom. In this case the splitting depends on the strength of the gravitomagnetic field. This phenomenon is the  \textbf{ gravitomagnetic Zeemaneffect}, \cite{Ma}. Note that the linear approximation works just fine, since the quadratic terms of the gauge field are already neglected  in the magnetic Zeemaneffect, \cite{St}. \newline\newline
In the next proposition we couple the two constant forces by a double deformation. 

\begin{proposition}
Let the deformation  matrix $B^{1}_{ij}$   be given as 
\begin{equation*}
 B^{1}_{ij}=m\,\varepsilon_{ijk}\Omega^{k},
\end{equation*}
where  $\Omega^{k}=(2GM/r_{hs}) \,{\omega}^{k}$ is the constant gravitomagnetic field for a hollow
spinning
sphere and let the deformation  matrix $B^{2}_{ij}$  be given as 
\begin{equation*}
 B^{2}_{ij}=- (e/2) \,\varepsilon_{ijk}B^{k},
\end{equation*}
where $B^k$ is a homogeneous magnetic field.\newline\newline
Then, the  deformed Hamiltonian
of the hydrogen atom, $\left((H^A)_{B^{1},\mathbf{X}}\right)_{B^{2},\mathbf{X}}$ becomes the Hamiltonian of the 
 Zeemaneffect generated by a
a constant external magnetic and gravitomagnetic 
field, i.e.
\begin{equation*}
 (H^A)_{B^{1}+B^{2},\mathbf{X}}\Psi=H_{GEM}\Psi,\qquad \Psi\in\mathcal{E.} 
\end{equation*}

\end{proposition}
\begin{proof} Since the deformation with the coordinate operator commutes with the potential term of the hydrogen atom, the proof is analog to the proof of Proposition  \ref{cnt3}.
\end{proof}
 \subsection{Arbitrary static gauge field}
 By only assuming the principle of Galilei-invariance the author in \cite{Ja} succeeded in
deriving the minimally coupled Hamiltonian plus a potential. Thus by demanding that our deformed
Hamiltonian respects the Galilei-invariance, we have to add a potential. This is justified since we showed that the deformation of a free
Hamiltonian induces  electromagnetism and gravitomagnetism. Moreover, in \cite{Ja} it was shown that the
gauge field and the potential can only depend on the coordinates. Therefore, our deformation
covers the whole range of abelian gauge fields, since we can induce such fields by choosing
functions of the coordinate operator to obtain an arbitrary gauge field. This fact is used in the
next sections to induce a variety of physical effects. \newline\newline
  In the next proposition we show the importance of adding
a potential to the Hamiltonian and for this purpose we need the four-momentum given as
\begin{equation*}
 P_{\mu}=\left( H_{0},P_{i}\right)=\left(-P_{k}P^{k}/(2m)+g\,\phi(\mathbf{X}), P_{i} \right),
\end{equation*}
where $\phi(\mathbf{X})$ is the electromagnetic potential  $\phi_E$ or the  gravitoelectromagnetic potential  $-\phi_G$.  Moreover, $g$ is a coupling constant given by $e$ in the electromagnetic case and by $-m$ in the gravitoelectromagnetic case.
\begin{proposition}\label{belqdef} Let the gauge field induced
by deformation of the Hamiltonian $(H_{0})_{B,\mathbf{Q}}$ (see Proposition \ref{wcfh}) be defined as
\begin{equation}\label{gfgf}
-g\,A_{r}(\mathbf{X}):=(B Q(\mathbf{X}))_k\partial_r Q^k(\mathbf{X}),
\end{equation}  
where $\mathbf{A}$ is the electromagnetic   or the  gravitoelectromagnetic vector potential.
Then, the commutator  of the deformed momentum  vectors gives the spatial part of the   field
strength tensor $F_{ij}$, 
\begin{align}\label{ncqp}
[P^{B,\mathbf{Q}} _{i} {,} P^{B,\mathbf{Q}} _{j}]&=  -ig\, F_{ij}(\mathbf{X}).
\end{align} 
Furthermore, the commutator of the deformed Hamiltonian with the deformed momentum gives the
Lorentz force $F^{L }_{j}$, i.e.
\begin{align}\label{ncqp1}
[P^{B,\mathbf{Q}} _{0} {,} P^{B,\mathbf{Q}} _{j}]&=   -g  [\phi(\mathbf{X}),P_{j}]
-  i \frac{g}{m}\, F_{jk}(\mathbf{X}) P^k_{B,\mathbf{Q}}= i F^{L }_{j}.
\end{align}
Moreover, by using the Jacobi identities for the commutator relations between the deformed momentum and Hamiltonian operators  the \textbf{homogeneous Maxwell-equations} emerge.

\end{proposition}
 \begin{proof}According to Proposition \ref{wcfh}, the deformation of
$H_{0}$ by an operator  $ \mathbf{Q} $ is given as follows 
\begin{align*}
(H_{0})_{B,\mathbf{Q}}\Psi
& 
= -\frac{1}{2m} \biggl(    
P_j +i(BQ)_s [Q^s,P_j]    \biggr ) \biggl(    
P^j +i(BQ)_r [Q^r,P^j]    \biggr ) \Psi +g\,\phi(\mathbf{X})\Psi\\&
=  \frac{1}{2m} \biggl(    
\mathbf{P}  -g\,\mathbf{A} (\mathbf{X})\biggr )^2 \Psi
+g\,\phi(\mathbf{X})\Psi,
  \end{align*}
where we used the fact that the potential can only be a function of the coordinate operator and thus remains invariant under deformation. Since in the former propositions we identified deformations induced by the coordinate
operator with the gauge field $-g\,\mathbf{A}$, the induced term in this deformation can  be
identified with a general static gauge field. Next, we calculate  Commutator (\ref{ncqp}). The deformed momentum operator is
given in Proposition \ref{defp} as
\begin{align*}
P^j_{B,\mathbf{Q}} \Psi
&=   (
P^{j}+i(BQ)_s [Q^s,P^j]   )\Psi  \\&
=   (
P^{j}   -g\,A^j (\mathbf{X})  )\Psi  ,
\end{align*}
where in the last lines we identified the gauge field by Equation (\ref{gfgf}).
Now by using the commutator relations $[X_{i},X_{j}]=0$  and the fact that the commutator
$[Q^s,P^j]$ is again only a function of the coordinate operator, we obtain the
following commutator relations for the deformed momentum operator
\begin{align*}
[P^k_{B,\mathbf{Q}} ,P^l_{B,\mathbf{Q}} ]&=
  -g\, [P^k, A^l (\mathbf{X})]  -k\leftrightarrow l\\&=-ig\,F^{kl} (\mathbf{X}).
\end{align*}
Now we can also rewrite the deformed Hamiltonian $(H_{0})_{B,\mathbf{Q}}  $ as 
\begin{align*}
 (H_{0})_{B,\mathbf{Q}}  =-\left(1/2m \right)\left(P^j_{B,\mathbf{Q}} P_j^{B,\mathbf{Q}
}\right)+g\,\phi(\mathbf{X}).
\end{align*}
This form of the deformed Hamiltonian simplifies  the calculation of Commutator (\ref{ncqp1}).
\begin{align*}
 [(H_{0})_{B,\mathbf{Q}}  ,P^k_{B,\mathbf{Q}} ]&=-\left(1/2m \right) [\left(P^j_{B,\mathbf{Q}} P_j^{B,
\mathbf{Q}
}\right),P^k_{B,\mathbf{Q}} ] +g\, [\phi(\mathbf{X}),P^k]\\&=  i \frac{g}{m}\, P_j^{B,\mathbf{Q}} F^{jk}(\mathbf{X})-i g\, \partial^k \phi(\mathbf{X})
\\&=  i g\,\left( E^k(\mathbf{X})- \varepsilon^{kjl}V^{B,\mathbf{Q}} _j B_l(\textbf{X})\right)
,
\end{align*}
In the last lines we used the commutator relation $[P_j^{B,\mathbf{Q}} ,F^{jk}(\mathbf{X})]=0$ and
the Heisenberg-equation to identify the velocity operator with the deformed momentum. Moreover, the fields $\mathbf{E}$ and $\mathbf{B}$ are the electromagnetic or the gravitoelectromagnetic fields, depending on the considered case. It is not
surprising that the Lorentz force is obtained by calculating the commutator between the deformed
Hamiltonian and the momentum, since it gives the equations of motion for the
deformed system. This in turn is properly identified with a particle coupled to an electromagnetic or gravitoelectromagnetic
force. 
\newline\newline
Next, we use the Jacobi identities for the commutators of the deformed momentum operator $P_{\mu}$ to obtain the homogeneous Maxwell-equations.
From the Jacobi identity for the spatial part we have
\begin{align*}
[P^{B,\mathbf{Q}} _{k} {,} [P^{B,\mathbf{Q}} _{i} {,} P^{B,\mathbf{Q}} _{j}]]+\text{cyclic}  &=-
 ig [P^{B,\mathbf{Q}} _{k} {,} F_{ij}(\mathbf{X})]+\text{cyclic}\\&=
g\,\partial_{k} F_{ij}(\mathbf{X})+\text{cyclic}\\&=0.
\end{align*} The last equation is the relativistic expression for the spatial part of the homogeneous Maxwell-equations. To obtain the homogeneous Maxwell-equations involving $F_{0j}=E_j$ we look at the other Jacobi identity of the deformed momentum, i.e. 
\begin{align}\label{me}
[P^{B,\mathbf{Q}} _{0} {,} [P^{B,\mathbf{Q}} _{i} {,} P^{B,\mathbf{Q}} _{j}]]
+[P^{B,\mathbf{Q}} _{i} {,} [P^{B,\mathbf{Q}} _{j} {,} P^{B,\mathbf{Q}} _{0}]]-i\leftrightarrow j
 =0.
\end{align}
Let us take a look at the first term, 
 \begin{align*}
[P^{B,\mathbf{Q}} _{0} {,} [P^{B,\mathbf{Q}} _{i} {,} P^{B,\mathbf{Q}} _{j}]]&= -ig
[P^{B,\mathbf{Q}} _{0} {,} F_{ij}(\mathbf{X}) ]\\&
= -g\,\partial_{0}F_{ij}(\mathbf{X}) -   \frac{g}{m}\, P^{B,\mathbf{Q}}_{k}\partial^k F_{ij}(\mathbf{X}),
\end{align*}
where  we used the Heisenberg-equation. The other two terms in Equation (\ref{me}) give 
 \begin{align*}
- i [P^{B,\mathbf{Q}} _{i} {,}\, \, F^L_j]-i\leftrightarrow j &= -g\,(\partial_i E_j(\mathbf{X}) -\partial_j E_i(\mathbf{X}) )- \frac{g}{m}\, P^k_{B,\mathbf{Q}}\left(\partial_i  F_{jk}(\mathbf{X})-\partial_j  F_{ik}(\mathbf{X})\right).
\end{align*}
By summing the two terms and using the spatial part of the homogeneous Maxwell-equations we obtain
 \begin{align*}
[P^{B,\mathbf{Q}} _{0} {,} [P^{B,\mathbf{Q}} _{i} {,} P^{B,\mathbf{Q}} _{j}]]+ \text{cyclic}
=-g\,(\partial_i E_j(\mathbf{X}) -\partial_j E_i(\mathbf{X}) )  -g\,\partial_{0}F_{ij}(\mathbf{X}).
\end{align*}

  \end{proof}

\begin{remark}
In \cite{FD} the author recollected an argument given by Feynman about  the most general force, compatible with the naive commutation relations
\begin{align*}
[X_{i},X_{j}]=0,\qquad m[X_{i},\dot{X}_{j}]=i\delta_{ij}.
\end{align*}
It turns out that by assuming a noncommutative structure for the velocity operators, one obtains the electromagnetic force. Furthermore, by using the Jacobi identities the Maxwell equations follow. Therefore, in a sense the deformations with warped convolutions reproduce the result of Feynman in a more sophisticated language, and moreover, the technique gives concrete representations of the operators that generate the electromagnetic and gravitoelectromagnetic fields. Thus, by the virtue of the deformation technique we have a deeper understanding of the surprising result of Feynman.  
\end{remark}$\,$\newline
A  crucial point is implied in the last proposition. The linear field equations of general relativity or the Maxwell equations   emerge from a well-defined deformation procedure. The emergence is owed to the Jacobi identities. This observation gives an insightful hint how to receive substantial field equations from a deformation procedure. 

 \subsection{Gravitomagnetic Moyal-Weyl plane}
By using the deformations techniques we are able to understand how to generate the Landau quantization. Thus, by using the same procedures  and by setting the  deformation parameter  equal to a constant \textbf{gravitomagnetic field} (see Lemma \ref{llq1}) we also obtain a Landau quantization in the gravitomagnetic case.\newline\newline
Analogously to the Landau quantization in the magnetic case, we can solve the eigenvalue equation of the  deformed Hamiltonian $(H_{0})_{B,X}$ with deformation matrix $B_{ij}=  m\varepsilon_{ijk}\,\Omega \,e^k $, i.e. with a constant homogeneous gravitomagnetic field in $x$-direction. The eigenvalue problem can be solved by diagonalizing the Hamiltonian and we obtain the following energy eigenvalues 
 \begin{equation*}
E_{B,n}=\frac{p_1^2}{2m}+\biggl(n+\frac{1}{2}\biggr)\omega_{B}, \qquad p_1\in\mathbb{R},\qquad
n\in \mathbb{N},
\end{equation*}
where the frequency of this harmonic oscillator is given by $\omega_{B}=2\Omega$. Therefore,   quantum mechanical particles in a constant gravitomagnetic field can only occupy orbits with discrete energy values. This effect is the gravitomagnetic analogue of the Landau quantization in the magnetic case.
\newline\newline 
In Section \ref{pmwfd} we identified the noncommutative coordinates of the Landau-quantization with  the deformed coordinate operator.  In the same manner  we can obtain a physical Moyal-Weyl plane for  the constant gravitomagnetic case. For this purpose, we set the
deformation matrix of  $\mathbf{X}^{\theta ,\mathbf{P}}$ equal to the inverse   of a  \textbf{constant gravitomagnetic field} times the coupling constant $m$. The deformed coordinate operators   $\mathbf{X}_{B^{-1},\mathbf{P}}$ are then equal to  the guiding center
coordinates of an electron in the lowest Landau
level and   are given by
\begin{equation*}
X^{B^{-1} ,\mathbf{P}}_{i}= X_{i}+\frac{1}{2}(B^{-1})_{ik}P^{k},
\end{equation*}
with  commutator relations
\begin{equation*}
 [X^{B^{-1} ,\mathbf{P}}_{i},X^{B^{-1} ,\mathbf{P}}_{j}]=  i(B^{-1})_{ij},
\end{equation*}
and deformation matrix $B_{ij}=m\varepsilon_{ijk}\Omega^{k}$. From the commutator relations  we have the following uncertainty relations 

\begin{equation*}
 (\Delta X^{B^{-1} ,\mathbf{P}}_{2}) (\Delta X^{B^{-1} ,\mathbf{P}}_{3})\geq \hbar/(m\Omega).
\end{equation*}\newline 
These uncertainty relations have the precise meaning that  coordinates of an electron can not be measured more accurately than the area $2\pi\,\hbar/(m\Omega)$.
This is a physical effect that we predicted from a deformation procedure and may be experimentally verified.
 
\section{Conclusion and Outlook} 
We obtained a variety of  physical effects, in a QM context, containing electromagnetism  and
gravitomagnetism. These effects were generated by the deformation procedure warped convolutions. Thus, in this sense those two
fundamental forces can be understood as deformations of free theories. The fundamental
\textbf{deformation parameters}, for those forces,  are given by the elementary electric charge $ {e}$
and by the gravitational constant ${G}$. Therefore, not only $\hbar$ and $c$ can be used to deform the
classical case (Galilei group) but also $e$ and $G$ play the role of deformation parameters
responsible for electromagnetism and gravitomagnetism. \newline\newline
The deformation also shed a new light on the dynamics of a quantum mechanical particle in the
presence of electromagnetic   and gravitomagnetic forces. Namely, it gives a \textbf{systematic derivation} of a non-relativistic Hamiltonian in the presence
of electromagnetic and gravitomagnetic effects. 
\newline \newline
Another striking implication of the deformation procedure is the deduction of a physical
Moyal-Weyl plane. This plane is generated   
from the gravitomagnetic field times the mass and thus the strength of noncommutativity of the
coordinates, in the lowest Landau level, is given by the inverse of the
constant $m\Omega/\hbar$. This effect was purely deduced from deformation and could be one of the
first effects that is theoretically predicted by deformation and  verified experimentally. This
would be a major physical breakthrough for deformation theory.
\newline \newline
To obtain electric and gravitoelectric fields in the framework of deformation, we would have to extend the definition of warped convolutions. For example, it is not possible to obtain the Stark effect from  deformation of the free Hamiltonian with warped convolutions. The reason herein lies in the fact, that the deformation leaves the spectrum of the operator invariant. Since the free Hamiltonian has a positive spectrum and the Hamiltonian of the Stark effect has the whole real line as spectrum,  a deformation respecting spectrum conditions  can not reproduce such an effect.  
 \newline \newline
Another line of work involves the extension of warped convolutions to a non-abelian setting. If this succeeds, we would be able to reproduce the weak and strong interaction as deformations. In this context, it seems intuitive to lift such deformations to a QFT setting.  Thus, recasting the fundamental forces as deformations and most likely simplifying calculations involving interactions.  
\newline \newline
These and many other lines of work, in deformation theory, are to this date open and provide a broad, interesting
and exciting area of research.
\section*{Acknowledgments} 
I would like to thank Prof. K. Sibold for   insightful remarks. I am particularly grateful
for the assistance given by Dr. G. Lechner. Moreover, I
thank S. Alazzawi, A. Andersson,  S. Pottel, J. Zschoche and Hyun Seok Yang   for many  discussions.     
I would like 
to express my  great appreciation to Dr. Z. Much for an extensive proofreading.

 \bibliographystyle{alpha}
\bibliography{allliterature}

\begin{thebibliography}{BEH08}

\bibitem[AC10]{AP}
Ronald~J. Adler and Pisin Chen.
\newblock {Gravitomagnetism in Quantum Mechanics}.
\newblock {\em Phys.Rev.}, D82:025004, 2010.

\bibitem[Ala]{A}
Sabina Alazzawi.
\newblock {Deformations of Fermionic Quantum Field Theories and Integrable
  Models}.
\newblock {\em Lett. Math. Phys. 103 (2013) 37-58}.

\bibitem[Alb12]{MUc}
Much Albert.
\newblock {Wedge-local quantum fields on a nonconstant noncommutative
  spacetime}.
\newblock {\em Journal of Mathematical Physics}, 53(8):082303, August 2012.

\bibitem[And13]{AA}
Andreas Andersson.
\newblock Operator deformations in quantum measurement theory.
\newblock {\em Letters in Mathematical Physics}, pages 1--16, November 2013.

\bibitem[Bal98]{BAL}
L.E. Ballentine.
\newblock {\em Quantum mechanics: a modern development}.
\newblock World Scientific Publishing Company Incorporated, 1998.

\bibitem[BEH08]{BEH}
J.~{Blank}, P.~{Exner}, and M.~{Havl{\'i}\v{c}ek}.
\newblock {\em {Hilbert Space Operators in Quantum Physics}}.
\newblock Springer, 2008.

\bibitem[Ber00]{Be}
R.A. Bertlmann.
\newblock {\em {Anomalies in quantum field theory}}.
\newblock International Series of Monographs on Physics. Oxford University
  Press, 2000.

\bibitem[BLS11]{BLS}
Detlev Buchholz, Gandalf Lechner, and Stephen~J. Summers.
\newblock {Warped Convolutions, Rieffel Deformations and the Construction of
  Quantum Field Theories}.
\newblock {\em Commun.Math.Phys.}, 304:95--123, 2011.

\bibitem[BS]{BS}
Detlev Buchholz and Stephen~J. Summers.
\newblock {Warped Convolutions: A Novel Tool in the Construction of Quantum
  Field Theories}.
\newblock {\em Quantum Field Theory and Beyond, pp. 107–121. World
  Scientific, Singapore.}

\bibitem[Dys90]{FD}
F.~J. Dyson.
\newblock {Feynman's proof of the Maxwell equations}.
\newblock {\em Am.J.Phys.}, 58:209--211, 1990.

\bibitem[Eza08]{EZ}
Z.F. Ezawa.
\newblock {\em Quantum Hall Effects: Field Theorectical Approach and Related
  Topics}.
\newblock World Scientific Publishing Company Incorporated, 2008.

\bibitem[GL07]{GL1}
Harald Grosse and Gandalf Lechner.
\newblock {Wedge-Local Quantum Fields and Noncommutative Minkowski Space}.
\newblock {\em JHEP}, 0711:012, 2007.

\bibitem[GL08]{GL2}
Harald Grosse and Gandalf Lechner.
\newblock {Noncommutative Deformations of Wightman Quantum Field Theories}.
\newblock {\em JHEP}, 0809:131, 2008.

\bibitem[GS68]{GS1}
I.M. Gel'fand and G.E. Shilov.
\newblock {\em Generalized functions}.
\newblock Generalized Functions. Academic Press, 1968.

\bibitem[H{\"o}r04]{H}
L.~H{\"o}rmander.
\newblock {\em The Analysis of Linear Partial Differential Operators II:
  Differential Operators with Constant Coefficients}.
\newblock Springer, 2004.

\bibitem[Jau64]{Ja}
J.M. Jauch.
\newblock {Gauge invariance as a consequence of Galilei-invariance for
  elementary particles}.
\newblock {\em Hel. Phys. Acta}, 37:284--292, 1964.

\bibitem[JM67]{LB}
L{\'e}vy-Leblond Jean-Marc.
\newblock {Nonrelativistic Particles and Wave Equations}.
\newblock {\em Commu. math. Phys.}, 6:286--311, 1967.

\bibitem[Lec12]{GL4}
Gandalf Lechner.
\newblock {Deformations of quantum field theories and integrable models}.
\newblock {\em Commun.Math.Phys.}, 312:265--302, 2012.

\bibitem[LST13]{GL5}
Gandalf Lechner, Jan Schlemmer, and Yoh Tanimoto.
\newblock {On the equivalence of two deformation schemes in quantum field
  theory}.
\newblock {\em Lett.Math.Phys.}, 103:421--437, 2013.

\bibitem[Mas00]{Ma}
Bahram Mashhoon.
\newblock {Gravitational couplings of intrinsic spin}.
\newblock {\em Class. Quant. Grav.}, 17:2399--2410, 2000.

\bibitem[MM11]{Mor}
Eric Morfa-Morales.
\newblock {Deformations of quantum field theories on de Sitter spacetime}.
\newblock {\em J.Math.Phys.}, 52:102304, 2011.

\bibitem[Rie93]{RI}
M.A. Rieffel.
\newblock {Deformation quantization for actions of $\mathbb{R}^{d}$}.
\newblock {\em Memoirs A.M.S.}, 506, 1993.

\bibitem[RS75]{RS1}
M.~Reed and B.~Simon.
\newblock {\em {Methods of Modern Mathematical Physics. 1. Functional
  Analysis}}.
\newblock {Gulf Professional Publishing}, 1975.

\bibitem[Str02]{St}
N.~Straumann.
\newblock {\em Quantenmechanik: ein Grundkurs {\"u}ber nichtrelativistische
  Quantentheorie ; mit 2 Tabellen}.
\newblock Springer-Lehrbuch. Springer Verlag, 2002.

\bibitem[Sza04]{SZ}
Richard~J. Szabo.
\newblock {Magnetic backgrounds and noncommutative field theory}.
\newblock {\em Int.J.Mod.Phys.}, A19:1837--1862, 2004.

\bibitem[Tes01]{T}
G.~Teschl.
\newblock {Mathematical Methods in Quantum Mechanics}.
\newblock {\em Amer. Math. Soc.}, 99, 2001.

\bibitem[Thi81]{Th}
W.E. Thirring.
\newblock {\em Quantum mechanics of atoms and molecules}.
\newblock A Course in Mathematical Physics. Springer-Verlag, 1981.

\bibitem[Wei72]{WS}
S.~Weinberg.
\newblock {\em Gravitation and cosmology: principles and applications of the
  general theory of relativity}.
\newblock Wiley, 1972.

\end{thebibliography}

 \end{document}